\documentclass[11pt,reqno]{amsart}
\usepackage{tikz}
\usepackage{subfig}
\usepackage{epstopdf}
\usepackage{epsfig}
\usepackage{math}
\usepackage{adjustbox}
\usepackage{rotating}
\usepackage{graphicx}

\renewcommand{\div}{\nabla\cdot}
\newcommand{\grad}{\nabla}

\renewcommand{\D}{\mathcal{D}}
\renewcommand{\E}{\mathcal{E}}
\renewcommand{\F}{\mathcal{F}}
\renewcommand{\H}{\mathcal{H}}

\renewcommand{\P}{\mathcal{P}}
\renewcommand{\L}{\mathcal{L}}
\renewcommand{\I}{\mathcal{I}}

\begin{document}
\title[]{Hessian transport Gradient flows}
\author[Li]{Wuchen Li}
\address{Department of Mathematics, University of California, Los Angeles.}
\email{wcli@math.ucla.edu}

\author[Ying]{Lexing Ying}
\address{Department of Mathematics, Stanford University and Facebook AI Research}
\email{lexing@stanford.edu}


\maketitle

\begin{abstract} 
We derive new gradient flows of divergence functions in the probability space embedded with a class of Riemannian metrics. The Riemannian metric tensor is built from the transported Hessian operator of an entropy function. The new gradient flow is a generalized Fokker-Planck equation and is associated with a stochastic differential equation that depends on the reference measure. Several examples of Hessian transport gradient flows and the associated stochastic differential equations are presented, including the ones for the reverse Kullback--Leibler divergence, $\alpha$-divergence, Hellinger distance, Pearson divergence, and Jenson--Shannon divergence.
\end{abstract}

\section{Introduction}

The de Bruijn identity plays crucial roles in information theory, probability, statistics, geometry, and machine learning
\cite{ChowLiZhou2018_entropy,CoverThomas1991_elements,CsiszarShields2004_information,Nelson2,S,ZozorBrossier2015_debruijn}. It states that the dissipation of the relative entropy, also known as the Kullback--Leibler (KL) divergence function, along the heat flow is equal to the relative Fisher information functional. This identity is important for many applications in Bayesian statistics and Markov chain Monte Carlo methods.

It turns out that there are two geometric structures in the probability space related to the de Bruijn identity. One is Wasserstein geometry (WG) \cite{Lafferty,Villani2009_optimal}, which refers to the heat flows or Gaussian kernels.  In \cite{JKO,otto2001}, it shows that the gradient flow of the negative Boltzmann Shannon entropy in WG is the heat equation. The de Bruijn identity can be understood as the rate of entropy dissipation within WG. The other one is information geometry (IG) \cite{IG,IG2}, which relates to the differential structures of the entropy. IG studies various families of Hessian geometry of entropy and divergence functions. Here, the Boltzmann-Shannon entropy, the Fisher-Rao metric, and the further induced KL divergence function are of particular importance. Besides these classical cases, one also studies generalized entropy and divergence functions, such as Tsallis
entropy and Tsallis divergence \cite{divergence,Tsallis1988}.

A natural question arises: {\em What are natural families of geometries in the probability space that connect entropy/divergence functions, heat flows, and the de Bruijn identity}?

\begin{figure}
  \centering \includegraphics[scale=0.1]{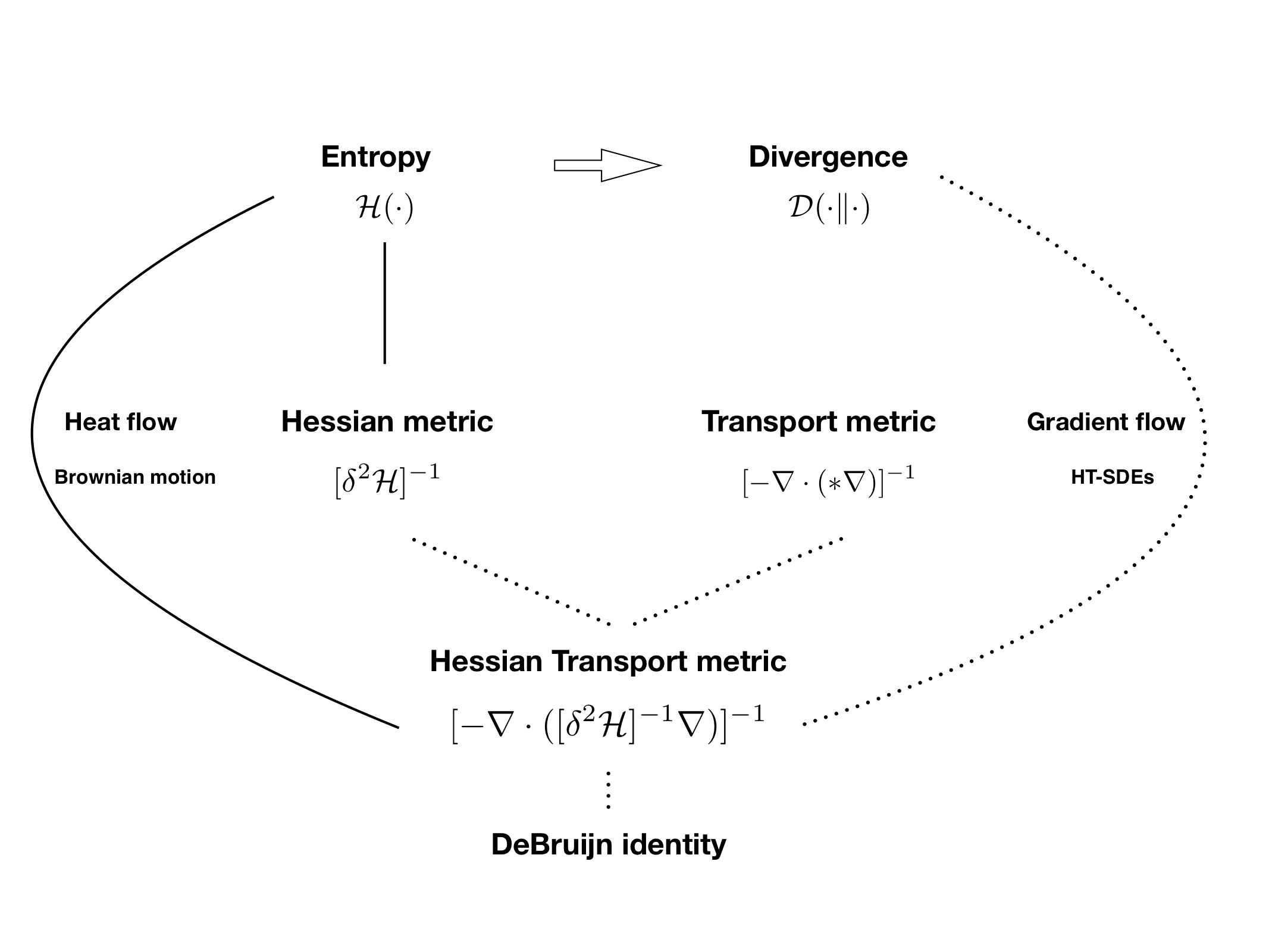}
  \caption{Derivation Diagram of the Hessian transport metric. The real lines are results in classical WG or IG communities. The dot lines are results derived in this paper. }
  \label{fig:results}
\end{figure}

In this paper, we positively answer this question by introducing a family of Riemannian metrics in
the probability space. Consider a compact space $\Omega$ and a positive smooth probability space
$\P(\Omega)$. For a strictly convex entropy function $\H\colon \P(\Omega)\rightarrow \mathbb{R}$, we
introduce a new Riemannian metric tensor $G_{\H}$ in the probability space
\begin{equation*}
  G_{\H}(\rho)^{-1}=[-\nabla\cdot([\delta^2\H(\rho)]^{-1}\nabla)],
\end{equation*}
where $\rho\in\P(\Omega)$ is a probability density function, $\delta^2\H(\rho)$ is the $L^2$ Hessian
operator of the entropy function, and finally $\nabla$ and $\nabla\cdot$ are gradient and divergence
operators on $\Omega$, respectively. We refer to Definition \ref{metric} for the formal
definition. Notice that the proposed metric involves both the Hessian geometry of $\H$ and
the transport metric (gradient and divergence operator on sample space). For this reason, it is called the {\em Hessian transport metric} (HT-metric)
(see Figure \ref{fig:results} for a schematic diagram).

As a simple but motivating example, the heat equation is the gradient flow of the entropy
function $\H$ under the HT-metric $G_\H$ induced by $\H$ itself:
\begin{equation*}
  \begin{split}
    \partial_t\rho=&-G_{\H}(\rho)^{-1}\delta \H(\rho)\\
    =&-\Big(-\nabla\cdot([\delta^2\H(\rho)]^{-1}\nabla \delta \H(\rho))\Big)\\
    =&\nabla\cdot([\delta^2\H(\rho)]^{-1}[\delta^2\H(\rho)]\nabla\rho)\\
    =&\nabla\cdot(\nabla \rho)=\Delta\rho.
  \end{split}
\end{equation*}

In a more general setting, we consider a family of entropy functions of form $\H_f(\rho)=\int_\Omega
f(\rho) dx$, where $f(\cdot)$ is convex, $f(1)=0$ and $f''$ is homogeneous of degree
$(-\gamma)$. For a fixed reference measure $\mu$, there is an associated divergence function
$\D_f(\rho\|\mu)=\int_\Omega f(\frac{\rho}{\mu})\mu dx$ for each $\H_f$. By considering the gradient
flow of $\D_f(\rho\|\mu)$ in $(\P(\Omega), G_{\H_f})$, we derive a generalized Fokker-Planck
equation
\[
\partial_t\rho(t,x)= \nabla\cdot\left(\mu(x)^{\gamma} \grad \left(\frac{\rho(t,x)}{\mu(x)}\right) \right),
\]
along with a stochastic differential equation for independent particle dynamics
\begin{equation*}
  dX_t = \gamma \mu(X_t)^{\gamma-2} \nabla \mu(X_t) dt + \sqrt{2 \mu(X_t)^{\gamma-1}} dB_t,
\end{equation*}
where $B_t$ is the standard Brownian motion. Such a SDE is called a {\em Hessian transport
  stochastic differential equation} (HT-SDE).

It is worth mentioning that two special cases of HT-metrics and their induced HT-SDEs are
particularly relevant.  When $\gamma=1$, $\H(\rho)$ is the Boltzmann-Shannon entropy, the HT-metric
is the usual Wasserstein-2 metric and the HT-SDE is the classical Langevin dynamics. When $\gamma=0$, $\H(\rho)$ is the Pearson divergence, the HT-metric is the $H^{-1}$ metric and the associated
HT-SDE is a diffusion process with zero drift. We refer to Table 1 for a summary of the results.

Currently, there are several efforts in combining both Wasserstein metric and information/Hessian
metric
\cite{LP,BauerJoshiModin2015_diffeomorphic,BauerModin2018_semiinvariant,CaoLuLu2018_exponential,WF2,Mielke,MMP,Wong}
from various perspectives. Within the Gaussian families, several extensions are studied in
\cite{Gaussian}.  Another example from the machine learning community is the Stein variational
gradient descent method \cite{Liu2017_steina,LiuWang2016_steina,LuLuNolen2019}. Here we introduce a
new geometry structure, which keeps heat flows as gradient flows of general entropy functions. We
emphasize the interaction relations between Hessian of entropy and transport metric. Our approach is
a natural extension to both IG and WG. It can also be viewed as a generalization for the field of
Wasserstein information geometry \cite{Li2018_geometrya,LiM,LiMontufar2018_ricci}.

The rest of this paper is organized as follows. In Section \ref{section2}, we define the Hessian
transport metric and show that the heat flow can be interpreted as the gradient flow of several
energy functions under appropriate HT-metrics. We then move on to derive, for general divergence
functions, the HT-metric gradient flows and the associated HT-SDEs. In Section \ref{section3}, we
introduce the Hessian transport distance (HT-distance) and derive the corresponding HT-geodesic
equation. Several numerical examples are given in Section \ref{section4}.

{
\begin{sidewaystable}
\vspace{15cm}
\centering
\scalebox{0.7}{
\begin{tabular}{|l|c|c|c|c|c|}
\hline 
Divergence & Inverse HT-Metric & Relative Fisher divergence & HT-Gradient flow & HT-SDE   & HT-geodesics\\
\hline
$\int_{\Omega}\rho\log{\frac{\rho}{\mu}}dx$ & $-\nabla\cdot(\rho\nabla)$& $\int_\Omega \left\|\nabla\log\frac{\rho}{\mu}\right\|^2\rho dx$ & $\partial_t\rho=\nabla\cdot(\rho\nabla \log \mu)+\Delta\rho$&
$dX_t=\nabla\log \mu dt+\sqrt{2}B_t$ &
$\begin{cases}
&\partial_t\rho+\nabla\cdot(\rho\nabla\Phi)=0\\
&\partial_t\Phi+\frac{1}{2}(\nabla\Phi, \nabla\Phi)=0
\end{cases}$\\
\hline
$-\int_{\Omega}\mu\log{\frac{\mu}{\rho}}dx$ & $-\nabla\cdot(\rho^2\nabla)$ &  $\int_\Omega \left\|\nabla\Big(\frac{\rho}{\mu}\Big)^{-1}\right\|^2\rho^2 dx$  & $\partial_t\rho=\div \left(\mu^2 \grad \left(\frac{\rho}{\mu}\right) \right)$&
$dX_t = 2 \grad \mu dt + \sqrt{2\mu} dB_t$&
$\begin{cases}
&\partial_t\rho+\nabla\cdot(\rho^2\nabla\Phi)=0\\
&\partial_t\Phi+(\nabla\Phi, \nabla\Phi)\rho=0
\end{cases}$\\    
\hline
$\int_\Omega \frac{4}{1-\alpha^2} (1-(\frac{\rho}{\mu})^{\frac{1+\alpha}{2}})\mu dx$
& $-\nabla\cdot(\rho^{\frac{3-\alpha}{2}}\nabla)$&$(\frac{2}{\alpha-1})^2\int_\Omega \left\|\nabla\Big(\frac{\rho}{\mu}\Big)^{\frac{\alpha-1}{2}}\right\|^2\rho^{\frac{3-\alpha}{2}} dx$ &$\partial_t\rho = \div \left(\mu^{(3-\alpha)/2} \grad \left(\frac{\rho}{\mu}\right) \right)$&
$dX_t = \frac{3-\alpha}{2} \mu^{\frac{-1-\alpha}{2}}\grad \mu dt + \sqrt{2\mu^{\frac{1-\alpha}{2}}} dB_t$&
$\begin{cases}
&\partial_t\rho+\nabla\cdot(\rho^{\frac{3-\alpha}{2}}\nabla\Phi)=0\\
&\partial_t\Phi+\frac{3-\alpha}{2}(\nabla\Phi, \nabla\Phi)\rho^{\frac{1-\alpha}{2}}=0
\end{cases}$
\\    
\hline
$\int_\Omega(\sqrt{\rho}-\sqrt{\mu})^2dx$ & $-\nabla\cdot(\rho^{\frac{3}{2}}\nabla)$ & 
$4\int_\Omega \left\|\nabla\Big(\frac{\rho}{\mu}\Big)^{-\frac{1}{2}}\right\|^2 \rho^{\frac{3}{2}} dx$&
$\partial_t\rho= \div \left(\mu^{3/2} \grad \left(\frac{\rho}{\mu}\right) \right)$&
$dX_t = \frac{3}{2} \mu^{-\frac{1}{2}} \grad \mu dt + \sqrt{2\mu^{\frac{1}{2}}} dB_t$ &
$\begin{cases}
&\partial_t\rho+\nabla\cdot(\rho^{\frac{3}{2}}\nabla\Phi)=0\\
&\partial_t\Phi+\frac{3}{2}(\nabla\Phi, \nabla\Phi)\rho^{\frac{1}{2}}=0
\end{cases}$
\\    
\hline
$\int_\Omega(\frac{\rho}{\mu}-1)^2\mu dx$ & $-\nabla\cdot(\nabla)$ & $\int_\Omega \left\|\nabla\Big(\frac{\rho}{\mu}\Big)\right\|^2dx$&
$\partial_t\rho=\nabla\cdot(\nabla(\frac{\rho}{\mu}))$&
$dX_t = \sqrt{2\mu^{-1}}dB_t$ &
$\begin{cases}
&\partial_t\rho+\Delta\Phi=0\\
&\partial_t\Phi=0
\end{cases}$\\    
\hline
$\int_\Omega \rho\log\frac{\rho}{\frac{1}{2}(\rho+\mu)}+\mu\log\frac{\mu}{\frac{1}{2}(\rho+\mu)} dx$ & $-\nabla\cdot(\rho(1+\rho)\nabla)$ &
$\int_\Omega \left\|\nabla\log(\frac{2\rho}{\rho+\mu})\right\|^2\rho(1+\rho)dx$&
$\partial_t\rho=\nabla\cdot\Big(\frac{(1+\rho)\mu^2}{\rho+\mu}\nabla(\frac{\rho}{\mu})\Big)
$&
- &
$\begin{cases}
&\partial_t\rho+\nabla\cdot(\rho(1+\rho)\nabla\Phi)=0\\
&\partial_t\Phi+\frac{1}{2}(\nabla\Phi,\nabla\Phi)(2\rho+1)=0
\end{cases}$\\    
\hline
\end{tabular}}
\label{table1}
\caption{In this table, we present the Hessian transport for KL divergence, reverse KL
  divergence, $\alpha$-divergence, Hellinger distance, Pearson divergence, and Jensen-Shannon divergence. In the case $\alpha=1$, the HT-metric recovers the Wasserstein-2 metric from the classical optimal transport theory. In the case of $\alpha=3$, the HT-metric recovers the $H^{-1}$ metric. }
\end{sidewaystable}}

\newpage
\section{Hessian transport Gradient flows}\label{section2}

In this section, we introduce the Hessian transport metrics and derive the gradient flows under these metrics.

\subsection{Motivations}\label{mot}
Consider a compact space $\Omega\subset\mathbb{R}^d$. Following the usual convention, we denote by
$\nabla$ and $\nabla\cdot$ the gradient and divergence operators in $\Omega$ and by $\delta$, $\|\cdot\|$ the Euclidean norm in $\mathbb{R}^d$, $\delta^2$ the first and the second $L^2$ variations. From now on, the boundary conditions on
$\Omega$ are given by either Neumann or periodic boundary conditions.

The heat equation
\begin{equation}\label{heat}
\partial_t\rho(t,x)=\Delta \rho(t,x),
\end{equation}
can be written in several equivalent ways as follows
\begin{equation*}
  \begin{aligned}
    \partial_t\rho(t,x)=&\nabla\cdot\Big( \nabla \rho(t,x)\Big)\\
    \partial_t\rho(t,x)=&\nabla\cdot\Big(\rho(t,x) \nabla \log\rho(t,x)\Big)\\
    \partial_t\rho(t,x)=&\nabla\cdot\left(\rho(t,x)^2\nabla \left(-\frac{1}{\rho(t,x)}\right) \right),
  \end{aligned}
\end{equation*}
where the following relation is used
\[
\nabla\rho=\rho\nabla\log\rho=\rho^2\nabla\left(-\frac{1}{\rho}\right).
\]

These formulas show that the heat flow has multiple gradient descent flow interpretations. Recall that a general gradient flow takes the form
\begin{equation*}
  \partial_t\rho=-G(\rho)^{-1}\delta\E(\rho),
\end{equation*}
where $\E(\cdot)$ is an energy function and the operator $G(\rho)\colon
C^{\infty}(\Omega)\rightarrow C^{\infty}(\Omega)$ represents the metric tensor. Under this framework, the heat equation can be interpreted in several ways:
\begin{itemize}
\item[(i)] Dirichlet energy formulation:
  \begin{equation*}
    \E(\rho)=\int_\Omega \|\nabla\rho(x)\|^2 dx,\quad\delta\E(\rho)=-\Delta \rho,\quad
    G(\rho)^{-1}=\mathbb{I},
  \end{equation*}
  where $\mathbb{I}\colon C^{\infty}(\Omega)\rightarrow C^{\infty}(\Omega)$ is the identity operator. Then 
  \begin{equation*}
    \partial_t\rho=-G(\rho)^{-1}\delta\E(\rho)=-(-\Delta\rho)=\Delta\rho.
  \end{equation*}
\item[(ii)] Boltzmann-Shannon entropy formulation:
  \begin{equation*}
    \E(\rho)=\int_\Omega \rho(x)\log\rho(x) dx, \quad \delta\E(\rho)=\log\rho+1, \quad
    G(\rho)^{-1}=-\nabla\cdot(\rho \nabla).
  \end{equation*}
  \begin{equation*}
    \partial_t\rho=-G(\rho)^{-1}\delta\E(\rho)=-(-\nabla\cdot(\rho\nabla\log\rho))=\Delta\rho.
  \end{equation*}
\item[(iii)] Cross entropy formulation:\footnote {Given $\mu\in C^\infty(\Omega)$, the cross entropy is defined as follows
\begin{equation*}
\mathcal{H}(\rho, \mu)=-\int_{\Omega}\mu(x)\log\rho(x)dx
\end{equation*}
Here we let $\mu(x)=1$, for all $x\in\Omega$.} 
  \begin{equation*}
    \E(\rho)=-\int_\Omega \log\rho(x)dx,\quad \delta\E(\rho)=-\frac{1}{\rho}, \quad
    G(\rho)^{-1}=-\nabla\cdot(\rho^2\nabla).
  \end{equation*}
  \begin{equation*}
    \partial_t\rho=
    -G(\rho)^{-1}\delta\E(\rho)=-\left(-\nabla\cdot\left(\rho^2\nabla\left(-\frac{1}{\rho}\right)\right)\right)=\Delta\rho.
  \end{equation*}
\end{itemize}
It is clear that the metric $G$ and the energy $\E$ need to be compatible in order to give rise the
heat equation. In fact, given a strictly convex energy function $\E$ in the probability space, it
induces a compatible metric operator
\[
G(\rho)^{-1} := \Big(-\nabla\cdot([\delta^2\E(\rho)]^{-1}\nabla)\Big),
\]
which combines both the transport operator (gradient, divergence operator in $\Omega$) and the $L^2$ Hessian operator of $\E$. Following this relation, the heat equation can be viewed as the gradient
flow of the energy $\E$ under the $\E$-induced metric operator. Below we include the calculations of
the above three cases for the sake of completeness.
\begin{itemize}
\item[(i)] Dirichlet energy formulation:
  \[
    \E(\rho)=\int_\Omega \|\nabla\rho(x)\|^2dx, \quad \delta^2\E(\rho)=-\Delta,
  \]
  \[
  G(\rho)^{-1}=\Big(-\nabla\cdot([\delta^2\E(\rho)]^{-1}\nabla)\Big)
  =\Big(-\nabla\cdot([-\Delta]^{-1}\nabla)\Big)\\
  =\mathbb{I}.
  \]
\item[(ii)] Boltzmann-Shannon entropy formulation:
  \[
  \E(\rho)=\int_\Omega \rho(x)\log\rho(x)dx,\quad  \delta^2\E(\rho)=\frac{1}{\rho},
  \]
  \[
  G(\rho)^{-1}=-\nabla\cdot([\delta^2\E(\rho)]^{-1}\nabla)  =-\nabla\cdot(\rho\nabla).
  \]
\item[(iii)] Cross entropy formulation:
  \[
  \E(\rho)=-\int_\Omega\log\rho(x)dx, \quad \delta^2\E(\rho)=\frac{1}{\rho^2}
  \]
  \[
  G(\rho)^{-1}=-\nabla\cdot([\delta^2\E(\rho)]^{-1}\nabla)\\
  =-\nabla\cdot(\rho^2\nabla).
  \]
\end{itemize}

\subsection{Hessian transport Gradient flows}
In this subsection, we will make the discussion in Subsection \ref{mot} precise. Consider the set of
smooth and strictly positive densities
\begin{equation*}
\P(\Omega)=\Big\{\rho \in C^{\infty}(\Omega)\colon \rho(x)>0,~\int_\Omega\rho(x)dx=1\Big\}. 
\end{equation*}
The tangent space of $\P(\Omega)$ at $\rho\in \P(\Omega)$ is given by 
\begin{equation*}
  T_\rho\P(\Omega) = \Big\{\sigma\in C^{\infty}(\Omega)\colon \int_\Omega\sigma(x) dx=0 \Big\}.
\end{equation*}  
For a strictly convex entropy function $\H\colon\P(\Omega)\rightarrow\mathbb{R}$, we first define
the following $\H$-induced metric tensor in the probability space.

\begin{definition}[Hessian transport metric tensor]\label{metric}
  The inner product $G_{\H}(\rho)\colon
  {T_\rho}\P(\Omega)\times{T_\rho}\P(\Omega)\rightarrow\mathbb{R}$ is defined as for any $\sigma_1$
  and $\sigma_2\in T_\rho\P(\Omega)$:
  \begin{equation*} G_{\H}(\rho)(\sigma_1,
    \sigma_2)=\int_{\Omega}\int_\Omega\Big(\sigma_1(x),
    \Big(-\nabla\cdot([\delta^2\H(\rho)]^{-1}\nabla)\Big)^{-1}(x,y)\sigma_2(y)\Big)dxdy,
  \end{equation*}
  where $[\delta^2\H(\rho)]^{-1}$ is the inverse of $L^2$ Hessian operator of $\H$,
  and $$\Big(-\nabla\cdot([\delta^2\H(\rho)]^{-1}\nabla)\Big)^{-1}\colon
  {T_\rho}\P(\Omega)\rightarrow{T_\rho}\P(\Omega)$$ is the inverse of weighted
  elliptic operator $-\nabla\cdot([\delta^2\H(\rho)]^{-1}\nabla)$.
\end{definition}
The proposed metric tensor is an extension of the Wasserstein metric. To see it, we represent the metric tensor into a cotangent bundle \cite{Li2018_geometrya,otto2001}.  Denote the space of potential functions on $\Omega$ by $\F(\Omega)$ and consider the quotient space $\F(\Omega)/\mathbb{R}$. Here each $\Phi\in \F(\Omega)/ \mathbb{R}$ is a function defined up to an additive constant.

We first show that $\F(\Omega)/\mathbb{R}$ is the cotangent bundle $T_\rho^*\P(\Omega)$. Consider
the identification map $\mathbf{V}\colon\F(\Omega)/\mathbb{R} \rightarrow T_\rho\P(\Omega)$ defined
by
\begin{equation*}
  \mathbf{V}(\Phi)=-\nabla\cdot([\delta^2\mathcal{H}(\rho)]^{-1}\nabla\Phi).
\end{equation*} 
At any $\rho$, define the elliptic operator
\begin{equation}\label{LP}
L_{\H,\rho} =-\nabla\cdot([\delta^2\H(\rho)]^{-1}\nabla).
\end{equation}
The uniform elliptic property of $L_{\H,\rho}$ guarantees that
$\mathbf{V}\colon\F(\Omega)/\mathbb{R}\rightarrow T_\rho\P(\Omega)$ is well-defined, linear, and
one-to-one. In other words, $\F(\Omega)/\mathbb{R}= T_\rho^*\P(\Omega)$. This identification further
induces the following inner product on $T_\rho\P(\Omega)$.

\begin{definition}[Hessian transport metric on the cotangent bundle]\label{d9}
The inner product $G_{\H}(\rho) :T_\rho\P(\Omega)\times T_\rho\P(\Omega)\rightarrow \mathbb{R}$ is
defined as for any two tangent vectors $\sigma_1=\mathbf{V}(\Phi_1)$ and
$\sigma_2=\mathbf{V}(\Phi_2)\in T_\rho\P(\Omega)$ 
\begin{equation*}\begin{split}
    G_{\H}(\rho)(\sigma_1, \sigma_2)=&\int_\Omega\sigma_1\Phi_2dx=\int_\Omega\sigma_2\Phi_1dx\\
    =&\int_\Omega\int_\Omega(\nabla\Phi_1(x), [\delta^2\H(\rho)]^{-1}(x,y)\nabla\Phi_2(y)) dxdy.
  \end{split} 
\end{equation*}
\end{definition}

Here the equivalence of Definition \ref{metric} and \ref{d9} is shown as follows. By denoting
$\sigma_i(x)=\mathbf{V}(\Phi_i)=L_{\H,\rho}\Phi_i$ for $i=1,2$, i.e.
\begin{equation*}
  \sigma_i(x)=-\nabla\cdot\left(\int_\Omega[\delta^2\H(\rho)]^{-1}(x,y)\nabla\Phi_i(y)dy\right)(x),
\end{equation*}
one has
\begin{equation*}
\begin{split}
  &\int_\Omega\int_{\Omega}(\nabla\Phi_1, [\delta^2\H(\rho)]^{-1}\nabla\Phi_2) dxdy = \int_\Omega \int_\Omega(\Phi_1, L_{\H,\rho}\Phi_2)dxdy \\
  =&\int_\Omega\int_\Omega\mathbf{V}(\Phi_1) L_{\H,\rho}^{-1}L_{\H,\rho}L_{\H,\rho}^{-1}\mathbf{V}(\Phi_2) dxdy
  =\int_\Omega\int_\Omega \sigma_1 L_{\H,\rho}^{-1}\sigma_2dxdy,
\end{split}
\end{equation*}
where in the the first equality we apply the integration by parts with respect to $\Omega$ using the boundary condition.
\begin{remark}
  In particular, if $\H(\rho)=\int_\Omega \rho(x)\log\rho(x) dx$, then
  $[\delta^2\H(\rho)]^{-1}=\rho$ and the Hessian transport metric takes the form
  \begin{equation*}
    G_{\H}(\rho)(\sigma_1,\sigma_2)=\int_\Omega (\nabla\Phi_1,\nabla\Phi_2)\rho dx,
  \end{equation*}
  with $\sigma_i=-\nabla\cdot(\rho\nabla\Phi_i)$, $i=1,2$. In this case, the Hessian transport metric is the
  Wasserstein-2 metric \cite{otto2001, Villani2009_optimal}.
\end{remark}

We are now ready to introduce the gradient flows in $(\P(\Omega), G_{\H})$.
\begin{lemma}[Hessian transport Gradient flow]
Given an energy functional $\E\colon \P(\Omega)\rightarrow \mathbb{R}$, the
gradient flow of $\E$ in $(\P(\Omega), G_{\H})$ is
\begin{equation*}
  \partial_t\rho(t,x)=\nabla\cdot
  \left(\int_{\Omega}[\delta^2\H(\rho)]^{-1}(x,y)\nabla\delta\E(\rho)(y)dy\right).
\end{equation*}
\end{lemma}
\begin{proof}
  The proof follows the definition. The Riemannian gradient in $(\P(\Omega), G_{\H})$ is defined as 
  \begin{equation}\label{GD}
    G_{\H}(\rho)(\sigma,
    \textrm{grad}_{\H}\E(\rho))=\int_{\Omega}\delta\E(\rho)(x)\sigma(x)dx,\quad\textrm{for
      any $\sigma(x)\in T_\rho\P(\Omega)$}.
  \end{equation}
  Denote 
  \begin{equation}\label{sigma}
    \sigma(x)=-\nabla\cdot\left(\int_\Omega [\delta^2\H(\rho)]^{-1}(x,y)\nabla\Phi(y)dy\right).
  \end{equation}
  Thus 
  \begin{equation*}
    \Phi(x)=\int_\Omega \Big(-\nabla\cdot([\delta^2\H(\rho)]^{-1}\nabla\Big)^{-1}(x,y)\sigma(y)dy.  
  \end{equation*}
  Notice that
  \begin{equation*}
    \begin{split}
      \textrm{L.H.S. of \eqref{GD}}=&G_{\H}(\rho)(\sigma, \textrm{grad}_{\H}\E(\rho))\\
      =&\int_{\Omega}\Big(\int_\Omega \Big(-\nabla\cdot([\delta^2\H(\rho)]^{-1}\nabla\Big)^{-1}(x,y)\sigma(y)dy\Big) \textrm{grad}_{\H}\E(\rho)(x)dx \\
      =&\int_\Omega \Phi(x)\textrm{grad}_{\H}\E(\rho)(x)dx,
    \end{split}
  \end{equation*}
  where we applies the definitions of the metric tensor and $\sigma$ in \eqref{sigma}. On the other
  hand,
  \begin{equation*}
    \begin{split}
    \textrm{R.H.S. of \eqref{GD}}=&\int_{\Omega}\delta\E(\rho)(x)\sigma(x)dx\\
    =&\int_\Omega\delta\E(\rho)(x)\Big(-\nabla\cdot(\int_\Omega[\delta^2\H(\rho)]^{-1}(x,y)\nabla\Phi(y)dy)\Big)dx\\
    =&\int_\Omega \int_\Omega \Big(\nabla\delta\E(\rho)(x),[\delta^2\H(\rho)]^{-1}(x,y)\nabla\Phi(y)\Big)dxdy\\
    =&\int_\Omega \Phi(y)\left(-\nabla\cdot\left(\int_\Omega [\delta^2\H(\rho)]^{-1}(x,y)\nabla\delta\E(\rho)(x)dx\right) \right)dy,
    \end{split}
  \end{equation*}
  where the second equality is obtained by integration by parts with respect to $x$ and the third
  equality holds by integration by parts with respect to $y$. Interchanging $x$ and $y$ in the R.H.S. and comparing the L.H.S. and R.H.S. of \eqref{GD} for any
  $\Phi$, we obtain the gradient operator 
  \begin{equation*}
    \textrm{grad}_{\H}\E(\rho)(x)=-\nabla\cdot\left(\int_\Omega[\delta^2\H(\rho)]^{-1}(x,y)\nabla\Phi(y)dy\right).    
  \end{equation*}
  Thus the Riemannian gradient flow in $(\P(\Omega), G_{\H})$ satisfies
  \begin{equation*}
    \partial_t\rho(t,x)=-\textrm{grad}_{\H}\E(\rho)(t,x)=
    \nabla\cdot\left(\int_\Omega[\delta^2\H(\rho)]^{-1}(x,y)\nabla\delta\E(\rho)(y)dy\right).
  \end{equation*}
\end{proof}

In particular, when $\E(\rho)=\H(\rho)=\int_\Omega f(\rho(x)) dx$, the gradient flow of $\H$ in
$(\P(\Omega), G_{\H})$ satisfies the heat equation \eqref{heat} because
\begin{equation*}
  \delta \H(\rho)(x)=f'(\rho)(x), \quad \delta^2\H(\rho)(x,y)=f''(\rho)(x)\delta_{x=y}
\end{equation*}
and
\[
\textrm{grad}_{\H}\H(\rho)(x)=
-\nabla\cdot\left( \frac{1}{f''(\rho)(x)} \nabla (f'(\rho)(x))\right)
=-\nabla\cdot\left(\frac{1}{f''(\rho)(x)}f''(\rho)(x)\nabla\rho(x)\right)  
=-\Delta\rho(x).
\]
The gradient flow of $\H(\rho)$ in $(\P(\Omega),G_{\H}(\rho))$ is then given by
\[
\partial_t\rho(t,x)=-\textrm{grad}_{\H}\H(\rho)(t,x)=-(-\Delta\rho(t,x))=\Delta\rho(t,x),    
\]
which is the heat equation as demonstrate in Subsection \ref{mot}. 

\subsection{Divergence and Hessian transport SDE}
By taking $\E$ to be the divergence function associated with the entropy $\H$, we derive here a
class of generalized Fokker-Planck equations as the gradient flows under the Hessian transport
metrics. In addition, we also give the associated Hessian transport stochastic differential
equations (HT-SDEs).

To the entropy function $\H_f(\rho)=\int_\Omega f(\rho(x)) dx$, we can associate a corresponding
divergence function:
\begin{equation*}
  \D_{f}(\rho\|\mu)=\int_\Omega f\left(\frac{\rho(x)}{\mu(x)}\right) \mu(x)dx.
\end{equation*}
Here $\rho$, $\mu\in\P(\Omega)$ and $f\colon \mathbb{R}\rightarrow\mathbb{R}$ is a convex function such that $f(1)=0$. In the literature, $\D_f(\cdot\|\cdot)$ is called the $f$-divergence function.

\begin{theorem}[Hessian transport stochastic differential equations]\label{thm1}
  Given a reference measure $\mu\in\P(\Omega)$, the gradient flow of $\D_{f}(\rho\|\mu)$ in $(\P(\Omega), G_{\H_f})$ satisfies
  \begin{equation}\label{FPE}
    \partial_t\rho(t,x)=\nabla\cdot \left( f''(\rho)(t,x)\nabla f'\left(\frac{\rho}{\mu}\right)(t,x) \right).    
  \end{equation}
  In addition, when $f''(\cdot)$ is homogeneous of degree $-\gamma$, i.e.,
  \[   f''(t)=f''(1) t^{-\gamma}.
  \] 
  The equation \eqref{FPE} can be simplified
  \begin{equation}\label{FPEGD}
    \partial_t\rho(t,x)= \nabla\cdot\left(\mu(x)^{\gamma} \grad \left(\frac{\rho(t,x)}{\mu(x)}\right) \right),
  \end{equation}
  and it is the Kolmogorov forward equation of the stochastic differential equation
  \begin{equation}\label{HT-SDE}
    dX_t = \gamma \mu(X_t)^{\gamma-2} \nabla \mu(X_t) dt + \sqrt{2 \mu(X_t)^{\gamma-1}} dB_t,   
  \end{equation}
  where $B_t$ is the standard Brownian motion in $\Omega$.
\end{theorem}
\begin{proof}
We first derive the gradient flow in $(\P, G_{\H_f})$. Notice that 
\begin{equation}\label{df}
\delta\D_f(\rho\|\mu)(x):=\frac{\delta}{\delta\rho(x)} \D_f(\rho\|\mu)= f'\left(\frac{\rho}{\mu}\right)(x),
\end{equation}
and the transport metric is 
\begin{equation*}
G_{\H_f}(\rho)=\Big(-\nabla\cdot(f''(\rho)^{-1}\nabla)\Big)^{-1}.  
\end{equation*}
Thus the gradient flow of $\D_{f}(\rho\|\mu)$ in $(\P(\Omega), G_{\H_f})$ satisfies 
\begin{equation*}
\begin{split}
  \partial_t\rho(t,x)=&-G_{\H_f}(\rho)^{-1}\delta \H_{f}(\rho\|\mu)
  =-\Big([-\nabla\cdot(f''(\rho)^{-1}\nabla)]^{-1}\Big)^{-1}f'\left(\frac{\rho}{\mu}\right)\\
  =&\nabla\cdot\left(f''(\rho)^{-1}\nabla f'\left(\frac{\rho}{\mu}\right)\right)
  = \nabla\cdot\left(f''(\rho)^{-1}f''\left(\frac{\rho}{\mu}\right)  \nabla\left(\frac{\rho}{\mu}\right) \right).
\end{split}
\end{equation*}
Notice $f''\left(\frac{\rho}{\mu}\right)=\mu^{\gamma} f''(\rho)$ with $\gamma\in \mathbb{R}$
due to the homogeneity assumption. Then the gradient flow \eqref{FPE} can be simplified to
\[
\partial_t\rho(t,x)= \nabla\cdot\left(\mu(x)^{\gamma} \grad \left(\frac{\rho(t,x)}{\mu(x)}\right) \right).
\]
This equation is a Kolmogorov forward equation (see for example
\cite{Oksendal2013_stochastic,Pavliotis2014}) for the density evolution $\partial_t \rho = L^* \rho$
with the forward operator given by $L^* = \div (\mu^{\gamma}\cdot \grad (\frac{1}{\mu}\cdot ))$. In
order to obtain the corresponding stochastic differential equation, we write down its adjoint
equation, i.e., the Kolmogorov backward equation $\partial_t u = Lu$ for functions on $\Omega$ with
$L = \left(\frac{1}{\mu}\cdot\right) \div \left( \mu^{\gamma}\cdot\grad\right)$. More precisely:
\[
\partial_t u = \left(\frac{1}{\mu}\cdot\right) \div \left( \mu^{\gamma} \cdot\grad u\right)
= \gamma (\mu(x)^{\gamma-2} \grad \mu(x)) \cdot \grad u(x) + \mu(x)^{\gamma-1} \Delta u(x).
\]
By identifying the drift coefficient before $\grad u(x)$ and the diffusion coefficient before
$\Delta u(x)$, we arrive at the corresponding stochastic differential equations
\begin{equation}
   dX_t = \gamma \mu(X_t)^{\gamma-2} \grad \mu(X_t) dt + \sqrt{2 \mu(X_t)^{\gamma-1}} dB_t,
   \label{HT-SDEgamma}
\end{equation}
which finishes the proof. 
\end{proof}

Following the gradient flow relation \eqref{FPEGD}, the reference measure $\mu$ is the invariant
measure for the HT-SDE \eqref{HT-SDE}. We next derive a generalized de Bruijn identity that
characterizes the dissipation of the divergence function along the gradient flow.
\begin{corollary}[Hessian transport de Bruijn identity]
  Suppose $\rho(t,x)$ satisfies \eqref{FPE}, then
  \begin{equation*}
    \frac{d}{dt}\D_{f}(\rho(t,\cdot)\|\mu)=-I_{f}(\rho(t,\cdot)\|\mu),
  \end{equation*}
  where the $f$-relative Fisher information functional $I_{f}(\rho\|\mu)$ is given by
  \begin{equation}\label{GF}
    I_{f}(\rho\|\mu)=\int_\Omega\left\|\nabla f'\left(\frac{\rho}{\mu}\right)\right\|^2f''(\rho)^{-1}dx.
  \end{equation}
\end{corollary}
\begin{proof}
  The proof follows the dissipation of energy along gradient flows in the probability space. Notice
  that
  \begin{equation*}
    \begin{split}
      \frac{d}{d t}\D_{f}(\rho(t,\cdot)\|\mu)=&-\int_{\Omega}\delta \D_{f}(\rho(t,\cdot)\|\mu)\partial_t\rho dx \\ 
      =&\int_\Omega\delta \D_{f}(\rho(t,\cdot)\|\mu)\nabla\cdot(f''(\rho)^{-1}\nabla\delta \D_{f}(\rho(t,\cdot)\|\mu))dx\\
      =&-\int_\Omega(\nabla\delta \D_{f}(\rho(t,\cdot)\|\mu), \nabla\delta \D_{f}(\rho(t,\cdot)\|\mu))f''(\rho)^{-1}dx\\
      =&-\int_\Omega \left\|\nabla f'(\frac{\rho}{\mu})\right\|^2f''(\rho)^{-1}dx,
    \end{split}
\end{equation*}
where the last equality holds by formula \eqref{df}.
\end{proof}
Consider the case $f(\rho)=\rho\log\rho$. The $f$-entropy is the negative Boltzmann-Shannon entropy
$\int_\Omega\rho\log\rho dx$ and the $f$-divergence is the usual relative entropy
\begin{equation*}
  \D_f(\rho\|\mu)=\int_\Omega \frac{\rho(x)}{\mu(x)} \log\frac{\rho(x)}{\mu(x)} \mu(x) dx =
  \int_\Omega \rho(x) \log\frac{\rho(x)}{\mu(x)} dx.
\end{equation*}
In this case, $\delta \D_f(\rho\|\mu)=\log\frac{\rho}{\mu}+1$, thus 
\begin{equation*}
  \frac{d}{dt}\D_f(\rho(t,\cdot)\|\mu)=-\int_\Omega \left\|\nabla \log\frac{\rho(t,x)}{\mu(x)}\right\|^2\rho(t,x)dx. 
\end{equation*}
Here we recover the classical result that the dissipation of the relative entropy is equal to the
negative relative Fisher information functional. Our result extends this relation to any
$f$-divergence functions. For this reason, $I_{f}$ in \eqref{GF} is called the {\em $f$-relative
  Fisher information functional}.

\begin{remark}
Here we demonstrate the relations between our approaches and the ones in literature
\cite{ZozorBrossier2015_debruijn, CFPE, CaoLuLu2018_exponential}.  The generalized de Bruijn identity
and $f$-relative Fisher information functional \eqref{GF} recovers exactly the ones in
\cite{ZozorBrossier2015_debruijn} when $\mu$ is a uniform measure. They differ from
\cite{ZozorBrossier2015_debruijn} when $\mu$ is a non-uniform reference measure. Our approach always
generalizes the entropy dissipation as the geometric dissipation as gradient flows of the
probability manifold $(\P(\Omega), G_{\H})$, while \cite{ZozorBrossier2015_debruijn} studies the
dissipation of relative entropy among two heat flows for two variables in the divergence
function. Our approach is also different from the one in \cite{CFPE}. We derive a class of
Fokker-Planck equation \eqref{FPEGD} with parameter $\gamma$, while \cite{CFPE} studies the
Fokker-Planck equation\eqref{FPEGD} with $\gamma=1$.  Lastly, our approach differs from
\cite{CaoLuLu2018_exponential}. While \cite{CaoLuLu2018_exponential} proposes a {\em
  reference-measure-dependent} metric under which the Fokker-Planck equation \eqref{FPEGD} with
$\gamma=1$ is the gradient flow of the Renyi entropy, our approach introduces a class of {\em
  reference-measure-independent} metrics. They only depend on the $L^2$ Hessian operator of the
convex entropy function and allow us to derive a new class of Fokker-Planck equations \eqref{FPEGD}.
\end{remark}

\subsection{Examples}
Below we consider a few special but important cases of $f$-divergences, and present the
$f$-divergence induced HT-SDE in Theorem \ref{thm1}.
\begin{example}[KL divergence HT-SDE]
  \[  f(\rho) = \rho\log \rho,\quad f'(\rho) = \log \rho + 1,\quad f''(\rho) = 1/\rho, \quad \gamma = 1.   \]
  The gradient flow, the HT-SDE, and the relative Fisher information functional are, respectively,
  \[  \partial_t\rho(t,x)= \div \left(\mu(x) \grad \left(\frac{\rho(t,x)}{\mu(x)}\right) \right),   \]
  \[   dX_t = \mu(X_t)^{-1} \grad \mu(X_t) dt + \sqrt{2} dB_t,  \]
  \[ \I_f(\rho\|\mu)=\int_\Omega \left\|\nabla\log\frac{\rho(x)}{\mu(x)}\right\|^2\rho(x) dx.  \]
\end{example}

\begin{example}[Reverse KL divergence HT-SDE]
  \[   f(\rho) = -\log \rho, \quad f'(\rho)=-1/\rho, \quad f''(\rho)=1/\rho^2, \quad \gamma =2.   \]
  The gradient flow, the HT-SDE, and the relative Fisher information functional are, respectively,
  \[   \partial_t\rho(t,x) = \div \left(\mu(x)^2 \grad \left(\frac{\rho(t,x)}{\mu(x)}\right) \right),   \]
  \[   dX_t = 2 \grad \mu(X_t) dt + \sqrt{2\mu(X_t)} dB_t,  \]
  \[  \I_f(\rho\|\mu)=\int_\Omega \left\|\nabla\Big(\frac{\rho(x)}{\mu(x)}\Big)^{-1}\right\|^2\rho(x)^2 dx. \]
\end{example}

\begin{example}[$\alpha$-divergence HT-SDE]
  \[   f(\rho) = \frac{4}{1-\alpha^2} (1-\rho^{\frac{1+\alpha}{2}}), \quad f'(\rho)=\frac{2}{\alpha-1} \rho^{\frac{\alpha-1}{2}}, 
  \quad f''(\rho)=\rho^{\frac{\alpha-3}{2}}, \quad \gamma =\frac{3-\alpha}{2}.
  \]
  The gradient flow, the HT-SDE, and the relative Fisher information functional are, respectively,
  \[  \partial_t\rho(t,x) = \div \left(\mu(x)^{(3-\alpha)/2} \grad \left(\frac{\rho(t,x)}{\mu(x)}\right) \right),   \]
  \[  dX_t = \frac{3-\alpha}{2} \mu(X_t)^{\frac{-1-\alpha}{2}} \grad \mu(X_t) dt + \sqrt{2\mu(X_t)^{\frac{1-\alpha}{2}}} dB_t,  \]
  \[ \I_f(\rho\|\mu)=\left(\frac{2}{\alpha-1}\right)^2\int_\Omega
  \left\|    \nabla\left(\frac{\rho(x)}{\mu(x)}\right)^{\frac{\alpha-1}{2}}    \right\|^2
  \rho(x)^{\frac{3-\alpha}{2}} dx. \]
\end{example}

\begin{example}[Hellinger distance HT-SDE]
  \[
  f(\rho) = (\sqrt{\rho}-1)^2,
  \]
  and it is a special case of $\alpha$-divergence with $\alpha=0$ and hence $\gamma=3/2$.
  The gradient flow, the HT-SDE, and the relative Fisher information functional are, respectively,
  \[  \partial_t\rho(t,x) = \div \left(\mu(x)^{3/2} \grad \left(\frac{\rho(t,x)}{\mu(x)}\right) \right),\]
  \[  dX_t = \frac{3}{2} \mu(X_t)^{-\frac{1}{2}} \grad \mu(X_t) dt + \sqrt{2\mu(X_t)^{\frac{1}{2}}} dB_t, \]
  \[
  \I_f(\rho\|\mu)=4\int_\Omega
  \left\|\nabla\left(\frac{\rho(x)}{\mu(x)}\right)^{-1/2}\right\|^2
  \rho(x)^{\frac{3}{2}} dx.
  \]
\end{example}

\begin{example}[Pearson divergence HT-SDE]
  \[
  f(\rho) = (\rho-1)^2,
  \]
  and it is a special case of $\alpha$-divergence with $\alpha=3$ and hence $\gamma=0$.
  The gradient flow, the HT-SDE, and the relative Fisher information functional are, respectively,
  \[  \partial_t\rho(t,x) = \div \left( \grad \left(\frac{\rho(t,x)}{\mu(x)}\right) \right),  \]
  \[  dX_t = \sqrt{2\mu(X_t)^{-1}}dB_t,    \]
  \[
  \I_f(\rho\|\mu)=\int_\Omega \left\|\nabla\Big(\frac{\rho(x)}{\mu(x)}\Big)\right\|^2dx.
  \]
\end{example}

\begin{example}[Jensen-Shannon divergence HT-SDE]
  \[
  f(\rho) = -(\rho+1)\log\frac{1+\rho}{2} + \rho\log \rho,\quad  f'(\rho) = -\log\frac{1+\rho}{2} + \log \rho,\quad  f''(\rho) = \frac{1}{\rho(1+\rho)}
  \]
  The above discussion does not apply since $f''(\rho)$ is not homogeneous in $\rho$. However, one
  can still obtain a gradient flow PDE
  \[
  \partial_t\rho(t,x) = \div \left( \frac{(1+\rho(t,x))\mu(x)^2}{(\rho(t,x)+\mu(x))} \grad \left(\frac{\rho(t,x)}{\mu(x)}\right) \right),
  \]
  and the $f$--relative Fisher information functional  
  \begin{equation*}
    \I_f(\rho\|\mu)=\int_\Omega \left\|\nabla\log(\frac{2\rho}{\rho+\mu})\right\|^2\rho(1+\rho)dx.
  \end{equation*}
\end{example}

\section{Hessian transport distance}\label{section3}

In this section, we introduce the Hessian transport distance in the probability space following the
metric tensor proposed in Definition \ref{metric} and derive the geodesic equations.

\begin{definition}[Hessian transport distance]
  Given a convex entropy function $\H$, the distance function
  $W_{\H}\colon\P(\Omega)\times\P(\Omega)\rightarrow \mathbb{R}$ between two densities $\rho^0$ and
  $\rho^1$ is 
  \begin{equation}\label{HT}
    W_{\H}(\rho^0,\rho^1) = \left( \inf_{v,\rho}\int_0^1\int_{\Omega}\int_{\Omega}
    \left(v(t,x), [\delta^2\H(\rho)]^{-1}(x,y)v(t,y) \right) dx dy dt \right)^{1/2},
  \end{equation}
  such that the infimum is taken among all density path $\rho\colon [0,1]\times\Omega\rightarrow
  \mathbb{R}$ and vector field $v\colon[0,1]\times \Omega\rightarrow T\Omega$, satisfying 
\begin{equation*}
  \partial_t\rho(t,x)+\nabla\cdot \left(\int_\Omega [\delta^2\H(\rho)]^{-1}(x,y)v(t,y)dy \right) =0,
\end{equation*}
with $\rho(0,x) = \rho^0(x)$ and $\rho(1,x)=\rho^1(x)$.
\end{definition}

We first illustrate that $W_{\H}$ is a Riemannian distance. For a fixed density $\rho(x)$, the {\em Hodge decomposition} for a vector function $v(x)$ is
\[
v(x)=\nabla\Phi(x)+u(x),
\]
where $\Phi\colon \Omega\rightarrow\mathbb{R}$ and $u\colon \Omega\rightarrow T_x\Omega$ is the divergence free vector for
$\rho$ in the following sense
\begin{equation}\label{df}
  \nabla\cdot \left(\int_\Omega [\delta^2\H(\rho)]^{-1}(x,y)u(y)dy \right)(x)=0.    
\end{equation}
Thus
\begin{equation*}
  \begin{split}
    &\int_\Omega\int_\Omega (v(x), [\delta^2\H(\rho)]^{-1}(x,y)v(y))dxdy\\
    =&\int_\Omega\int_\Omega (\nabla\Phi(x), [\delta^2\H(\rho)]^{-1}(x,y)\nabla\Phi(y))dxdy+
    \int_\Omega\int_\Omega (u(x), [\delta^2\H(\rho)]^{-1}(x,y)u(y))dxdy\\
    &+2\int_\Omega\int_\Omega (u(x), [\delta^2\H(\rho)]^{-1}(x,y)\nabla\Phi(y))dxdy\\
        =&\int_\Omega\int_\Omega (\nabla\Phi(x), [\delta^2\H(\rho)]^{-1}(x,y)\nabla\Phi(y))dxdy+
    \int_\Omega\int_\Omega (u(x), [\delta^2\H(\rho)]^{-1}(x,y)u(y))dxdy\\
    \\
    \geq &\int_\Omega\int_\Omega (\nabla\Phi(x), [\delta^2\H(\rho)]^{-1}(x,y)\nabla\Phi(y))dxdy,
  \end{split}
\end{equation*}
where the second equality uses the divergence free relation \eqref{df}. Thus the minimization
problem \eqref{HT} is same as the one over variable $(\rho(t,x), \Phi(t,x))$, where $\Phi(t,x)$ is
the first part of the Hodge decomposition of $v(t,x)$. By denoting
$\partial_t\rho(t,x)=L_{\H,\rho}\Phi(t,x)$ with $L_{\H,\rho}$ defined in \eqref{LP}, we arrive at
\begin{equation*}
  \begin{split}
    G_{\H}(\rho)(\partial_t\rho, \partial_t\rho)=&\int_\Omega\int_\Omega \Big(\partial_t\rho(t,x), L_{\H,\rho}^{-1} \partial_t\rho(t,y)\Big)dxdy\\
    =& \int_\Omega\int_\Omega \Big(L_{\H,\rho}\Phi(t,x), L_{\H,\rho}^{-1} L_{\H,\rho}\Phi(t,y)\Big) dxdy\\
    =&\int_\Omega\int_\Omega (\Phi(t,x),L_{\H,\rho}\Phi(t,y))dxdy.
\end{split}    
\end{equation*}
Thus the distance function defined in \eqref{HT} can be formulated as 
\begin{equation*}
  \Big(W_{\H}(\rho^0,\rho^1)\Big)^2=\inf_{\rho\colon[0,1]\rightarrow\P(\Omega)}\Big\{\int_0^1
  G_{\H}(\rho)(\partial_t\rho, \partial_t\rho) dt \colon \textrm{$\rho^0$, $\rho^1$ fixed} \Big\}.
\end{equation*}
This is exactly the geometric action functional in $(\P(\Omega), G_{\H})$ and therefore $W_{\H}$ is
a Riemannian distance on $\P(\Omega)$.

Next, we prove that the distance is well-defined and derive the formulations of geodesics
equations.
\begin{theorem}[Hessian transport geodesic (HT-geodesic)]
  The Hessian transport distance is well-defined in $\P(\Omega)$,
  i.e. $W_{\H}(\rho^0,\rho^1)<+\infty$. The geodesic equation is
  \begin{equation}\label{geo}
    \left\{
    \begin{aligned}
      &\partial_t\rho(t,x)+\nabla\cdot (\int_\Omega [\delta^2\H(\rho)]^{-1}(x,y)\nabla\Phi(t,y)dy)=0\\
      &\partial_t\Phi(t,x)+\frac{1}{2}\delta_{\rho}\int_\Omega\int_\Omega \nabla\Phi(t,x), [\delta^2\H(\rho)]^{-1}(x,y)\nabla\Phi(t,y)dxdy=0.
    \end{aligned}\right.
  \end{equation}
  If $\H(\rho)=\mathcal{H}_f(\rho)=\int_\Omega f(\rho)(x)dx$, the geodesic equation simplifies to
  \begin{equation*}
    \left\{\begin{aligned}
    &\partial_t\rho+\nabla\cdot(f''(\rho)^{-1}\nabla\Phi)=0\\
    &\partial_t\Phi-\frac{1}{2}(\nabla\Phi, \nabla\Phi)\frac{f'''(\rho)}{f''(\rho)^2}=0.
    \end{aligned}
    \right.
  \end{equation*}
\end{theorem}
\begin{proof}
We first prove that the distance function is well-defined. First, by denoting
$m(t,x)=[\delta^{2}\H(\rho)]^{-1}v(t,x)$, we can rewrite the minimization problem \eqref{HT} as
\begin{equation}\label{miniflux}
  W_{\H}(\rho^0,\rho^1)^2=\inf_{\rho, m}\int_0^1\int_{\Omega}\int_{\Omega}(m(t,x), \delta^2\H(\rho)(x,y)m(t,y))dxdydt
\end{equation}
along with the constraints
\begin{equation*}
  \partial_t\rho(t,x)+\nabla\cdot m(t,x)=0,
\end{equation*}
with fixed initial and terminal densities $\rho^0$, $\rho^1$. We show that there exists a feasible
path for any $\rho^0$, $\rho^1\in \mathcal{P}(\Omega)$. Notice that
$\min\{\min_{x\in\Omega}\rho^0,\min_{x\in\Omega}\rho^1\}>0$. We construct a path
$\bar\rho(t,x)=(1-t)\rho^0+t\rho^1$, where $t\in[0,1]$. Thus $\bar\rho(t,x)\in\mathcal{P}(\Omega)$
and $\min_{{t,x}\in\Omega}\bar\rho(t,x)>0$. Construct a feasible flux function $\bar
m(t,x)=\nabla\bar\Phi(t,x)$, with $\bar\Phi(t,x)=-\Delta^{-1}\partial_t\bar\rho(t,x)\in
C^\infty(\Omega)$. Thus
\begin{equation*}
\int_0^1\int_{\Omega}\int_{\Omega}(\nabla\bar\Phi(t,x), [\delta^2\H(\rho)]^{-1}(x,y)\nabla\bar\Phi(t,y))dxdydt<\infty,
\end{equation*}
Then $(\bar\rho(t,x), \bar m=\nabla\bar\Phi(t,x))$ is a feasible path for minimization problem
\eqref{miniflux} for any $\rho^0$, $\rho^1\in\mathcal{P}(\Omega)$.

We next derive the geodesic equation within $\P(\Omega)$. The first step is to write down the
Lagrangian multiplier $\Phi(t,x)$ of the continuity equation $\partial_t\rho+\nabla\cdot m=0$
\begin{equation*}
  \begin{split}
    \L(m,\rho,\Phi)=&\frac{1}{2}\int_0^1\int_\Omega\int_\Omega (m(t,x), \delta^2\H(\rho)(x,y)m(t,y))dxdydt\\
    &+\int_0^1\int_\Omega \Phi(x)(\partial_t\rho(t,x)+\nabla\cdot m(t,x))dx. 
  \end{split}
\end{equation*}
At $\rho\in\P(\Omega)$, $\delta_\rho\L=0$, $\delta_m\L=0$, and $\delta_\Phi\L=0$,
we know that the minimizer satisfies
\begin{equation*}
  \left\{
  \begin{aligned}
    \int_\Omega \delta^2\H(\rho)(x,y)m(t,y)dy=\nabla\Phi(t,x),\\
    \frac{1}{2}\delta_{\rho}\int_\Omega\int_\Omega (m(t,x), \delta^2\H(\rho)(x,y)m(t,y))dxdy=\partial_t\Phi(t,x),\\
    \partial_t\rho(t,x)+\nabla\cdot m(t,x)=0.
  \end{aligned}\right.
\end{equation*}
Finally, by denoting $m(t,x)=\int_\Omega\delta^2\H(\rho)^{-1}(x,y)\nabla\Phi(t,y)dy$ and using the
fact
\begin{equation*}
  \delta_\rho [\delta^{2}\H(\rho)]=-[\delta^{2}\H(\rho)]^{-1}\delta_\rho [\delta^{2}\H(\rho)^{-1}]   [\delta^{2}\H(\rho)]^{-1},
\end{equation*}
we derive the geodesic equation \eqref{geo}. 
\end{proof}
\begin{remark}
  Consider the special case that $\H(\rho)$ is the $f$-entropy with $f''(\cdot)$ homogeneous of
  degree $-\gamma$. The objective function
  \[
  \int_0^1 \int_\Omega\int_\Omega (m(t,x), \delta^2\H(\rho)(x,y)m(t,y))dxdydt
  =\int_0^1 \int_\Omega \frac{\|m(t,x)\|^2}{\rho(t,x)^\gamma} dx dt
  \]
  is convex jointly in $(m,\rho)$ if and only if $\gamma\in [0,1]$. As two special cases, the
  proposed minimal flux minimization \eqref{miniflux} for the optimal KL ($\gamma=1$) and the
  Pearson ($\gamma=0$) Hessian transport are convex.
\end{remark}

\subsection{Examples}
Here we list the geodesic equation for the $f$-divergence functions.
\begin{example}[KL divergence HT-geodesic]
  \[
  f(\rho) = \rho\log \rho,\quad f'(\rho) = \log \rho + 1,\quad f''(\rho) = \frac{1}{\rho}, \quad f'''(\rho)=-\frac{1}{\rho^2}.
  \]
  Thus $f'''(\rho)/f''(\rho)^2=-1$ and the geodesic equation is
\begin{equation*}
  \left\{\begin{aligned}
  &\partial_t\rho+\nabla\cdot(\rho\nabla\Phi)=0\\
  &\partial_t\Phi+\frac{1}{2}(\nabla\Phi, \nabla\Phi)=0.
  \end{aligned}
  \right.
\end{equation*}
This is the classical geodesic equation in Wasserstein geometry, including both the continuity
equation and the Hamilton-Jacobi equation.
\end{example}
\begin{example}[Reverse KL divergence HT-geodesic]
  \[
  f(\rho) = -\log \rho, \quad f'(\rho)=-\frac{1}{\rho}, \quad f''(\rho)=\frac{1}{\rho^2}, \quad f'''(\rho)=-\frac{2}{\rho^3}.
  \]
  Thus $f'''(\rho)/f''(\rho)^2=-2\rho$, then the geodesic equation is 
  \begin{equation*}
    \left\{\begin{aligned}
    &\partial_t\rho+\nabla\cdot(\rho^2\nabla\Phi)=0\\
    &\partial_t\Phi+(\nabla\Phi, \nabla\Phi)\rho=0.
    \end{aligned}
    \right.
  \end{equation*}    
\end{example}

\begin{example}[$\alpha$-divergence HT-geodesic]
  \[
  f(\rho) = \frac{4}{1-\alpha^2} (1-\rho^{\frac{1+\alpha}{2}}), \quad f'(\rho)=\frac{2}{\alpha-1} \rho^{\frac{\alpha-1}{2}}, 
  \quad f''(\rho)=\rho^{\frac{\alpha-3}{2}},\quad f'''(\rho)=\frac{\alpha-3}{2}\rho^{\frac{\alpha-5}{2}}.
  \]
  Thus $f'''(\rho)/f''(\rho)^2=\frac{\alpha-3}{2}\rho^{\frac{1-\alpha}{2}}$, then the geodesic equation is
  \begin{equation*}
    \left\{\begin{aligned}
    &\partial_t\rho+\nabla\cdot(\rho^{\frac{3-\alpha}{2}}\nabla\Phi)=0\\
    &\partial_t\Phi+\frac{3-\alpha}{2}(\nabla\Phi, \nabla\Phi)\rho^{\frac{1-\alpha}{2}}=0.
    \end{aligned}
    \right.
  \end{equation*}    
\end{example}

\begin{example}[Hellinger distance HT-geodesic]
  \[
  f(\rho) = (\sqrt{\rho}-1)^2,
  \]
  and it is a special case of $\alpha$-divergence with $\alpha=0$. Hence the geodesic equation takes the form
  \begin{equation*}
    \left\{\begin{aligned}
    &\partial_t\rho+\nabla\cdot(\rho^{\frac{3}{2}}\nabla\Phi)=0\\
    &\partial_t\Phi+\frac{3}{2}(\nabla\Phi, \nabla\Phi)\rho^{\frac{1}{2}}=0.
    \end{aligned}
    \right.
  \end{equation*}  
\end{example}

\begin{example}[Pearson divergence HT-geodesic]
  \[
  f(\rho) = (\rho-1)^2,
  \]
  and it is a special case of $\alpha$-divergence with $\alpha=3$. 
  Hence the geodesic equation is
  \begin{equation*}
    \left\{\begin{aligned}
    &\partial_t\rho+\Delta\Phi=0\\
    &\partial_t\Phi=0.
    \end{aligned}
    \right.
  \end{equation*}   
  This geodesic equation satisfies $\frac{\partial^2}{\partial t^2}\rho(t,x)=0$, which implies that
  \[
  \rho(t,x)=(1-t)\rho^0(x)+t\rho^1(x).
  \]
  It states that the geodesic equation in optimal Hellinger distance transport metric is a straight
  line in the probability space.
\end{example}

\begin{example}[Jensen-Shannon divergence HT-geodesic]
  \[
  f(\rho) = -(\rho+1)\log\frac{1+\rho}{2} + \rho\log \rho,\quad
  f''(\rho) = \frac{1}{\rho(1+\rho)},\quad f'''(\rho)=-\frac{2\rho+1}{(\rho+1)^2\rho^2} 
  \]
  Thus ${f'''(\rho)}/{f''(\rho)^2}=-(2\rho+1)$. Hence the geodesic equation is 
  \begin{equation*}
    \left\{    \begin{aligned}
      &\partial_t\rho+\nabla\cdot(\rho(1+\rho)\nabla\Phi)=0\\
      &\partial_t\Phi+\frac{1}{2}(\nabla\Phi,\nabla\Phi)(2\rho+1)=0.        
    \end{aligned}\right.
  \end{equation*}
\end{example}

\section{Numerical examples}\label{section4}
In this section, we demonstrate the properties of the newly derived equations with several
examples. Since the gradient flow equations are linear, the flow dynamics are governed mostly by
the spectrum of the Kolmogorov forward and backward operators. Here, we consider the simple setting
of $\Omega$ equal to the unit interval $[0,1]$ with periodic boundary condition. The PDEs are
numerically discretized with a finite element method with a uniform discretization.
  
\begin{figure}[h!]
  \centering
  \begin{tabular}{ccc}
    \includegraphics[scale=0.1]{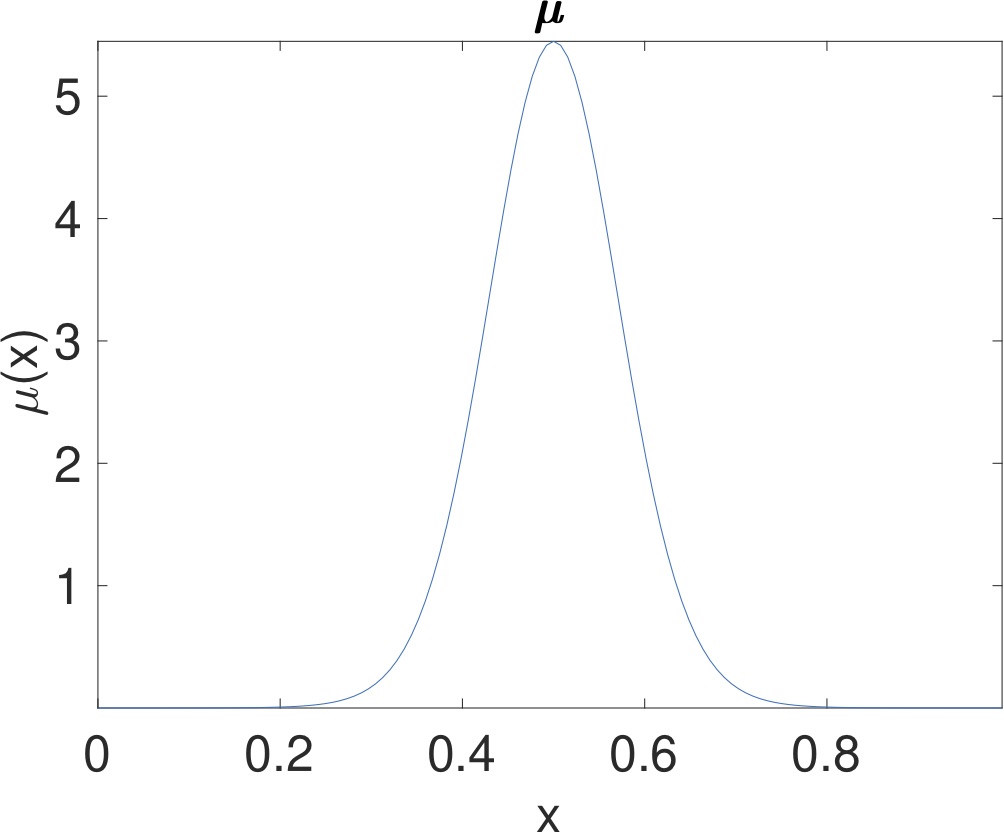}&
    \includegraphics[scale=0.1]{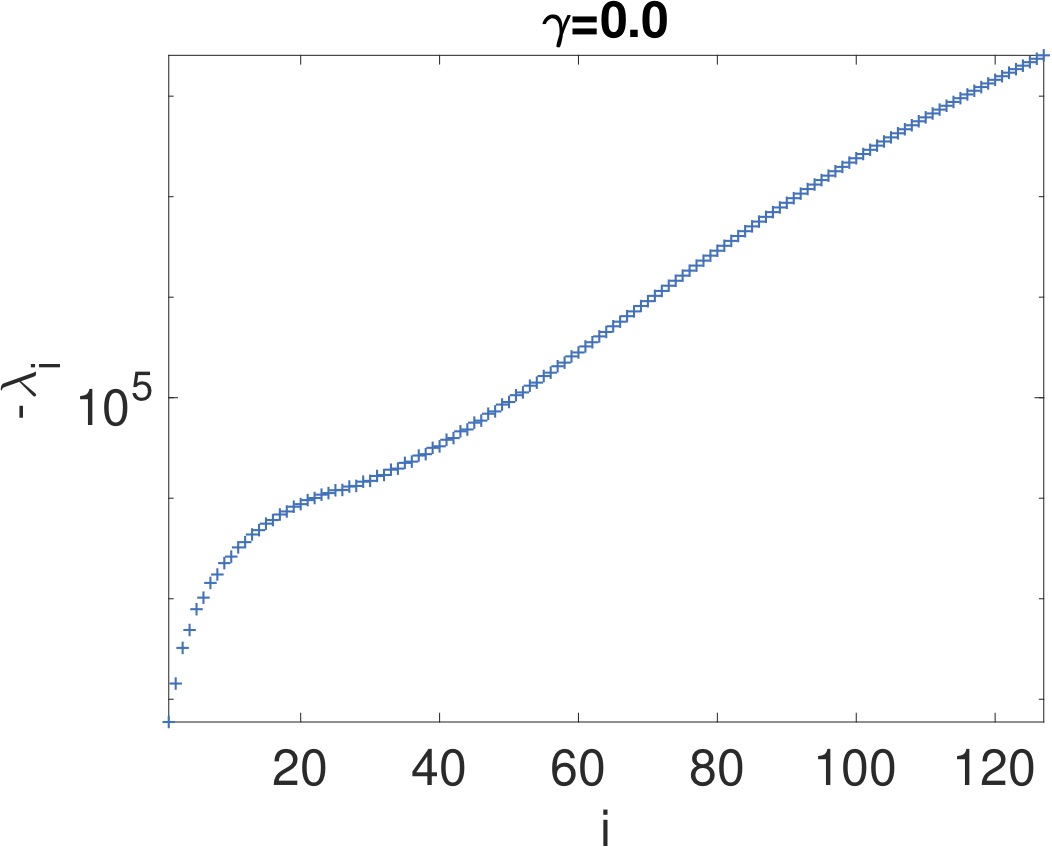}&
    \includegraphics[scale=0.1]{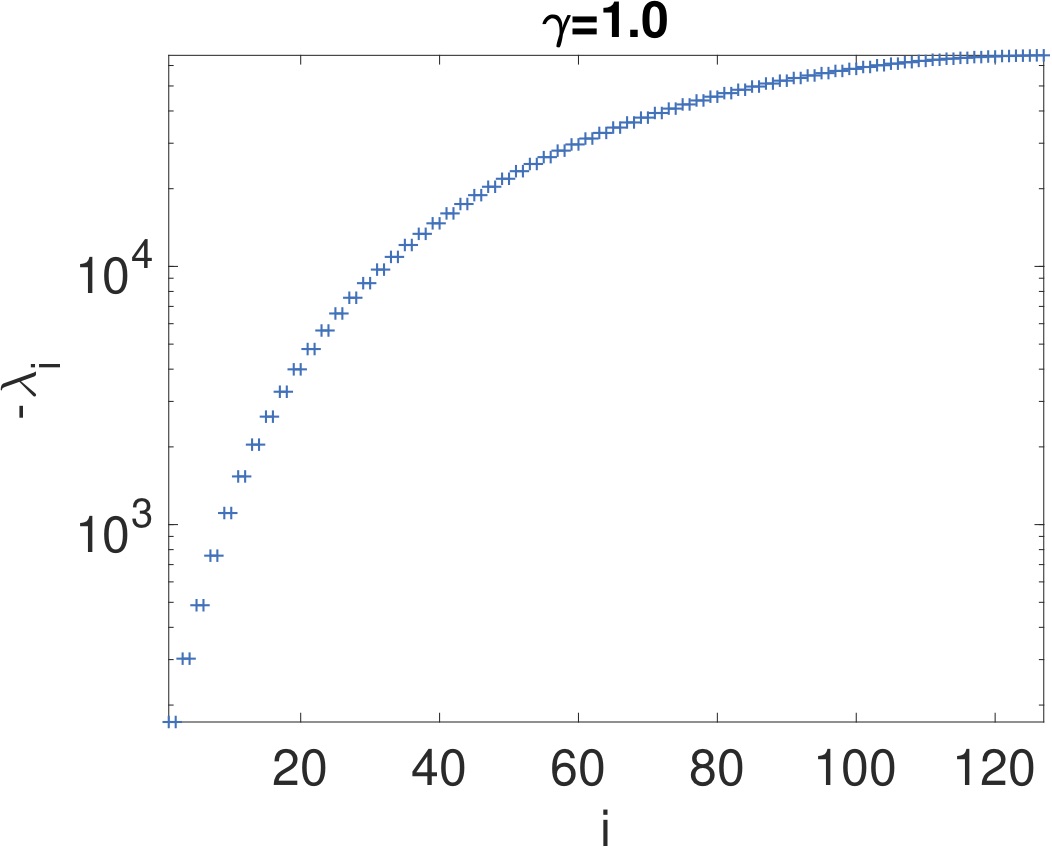}\\
    (a)&(b)&(c)\\
    \includegraphics[scale=0.1]{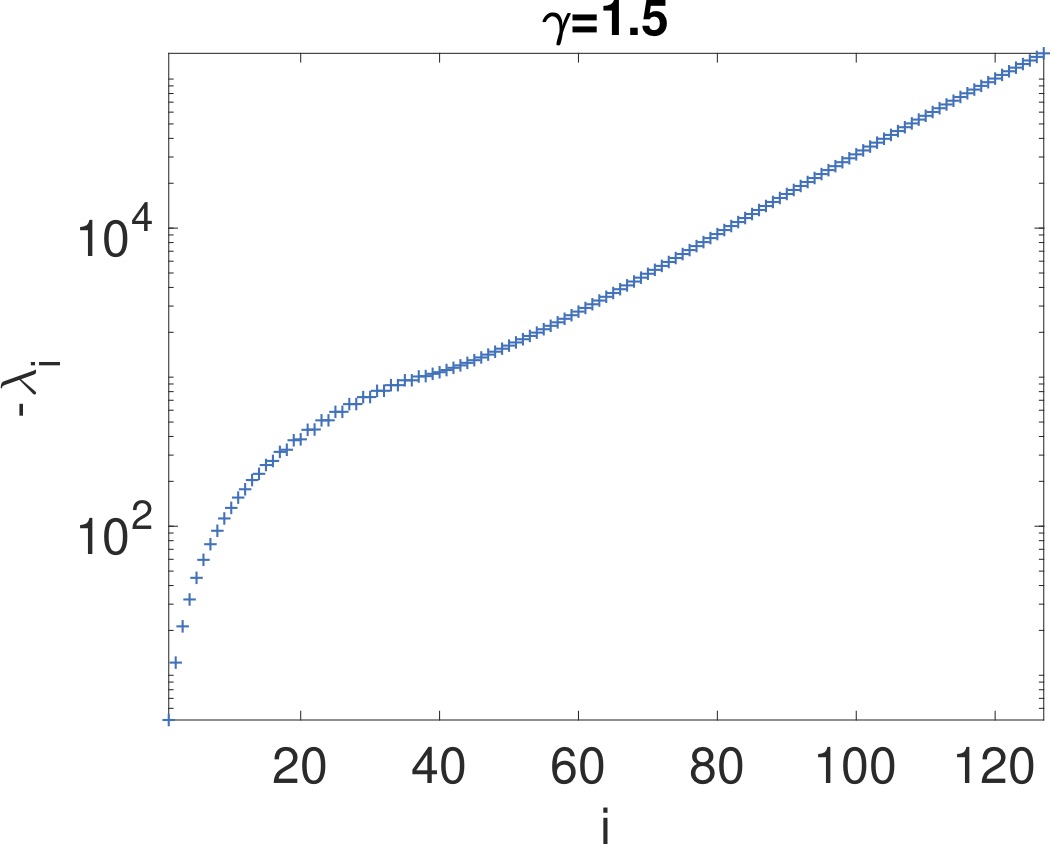}&
    \includegraphics[scale=0.1]{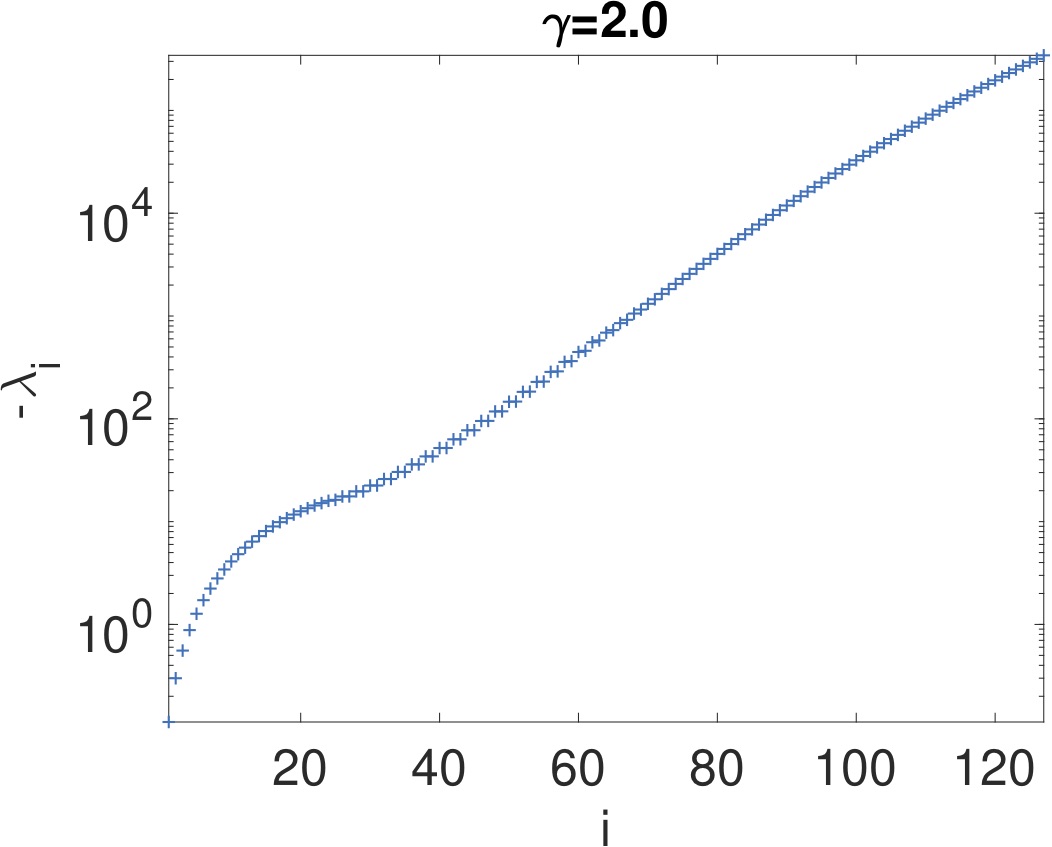}&
    \includegraphics[scale=0.1]{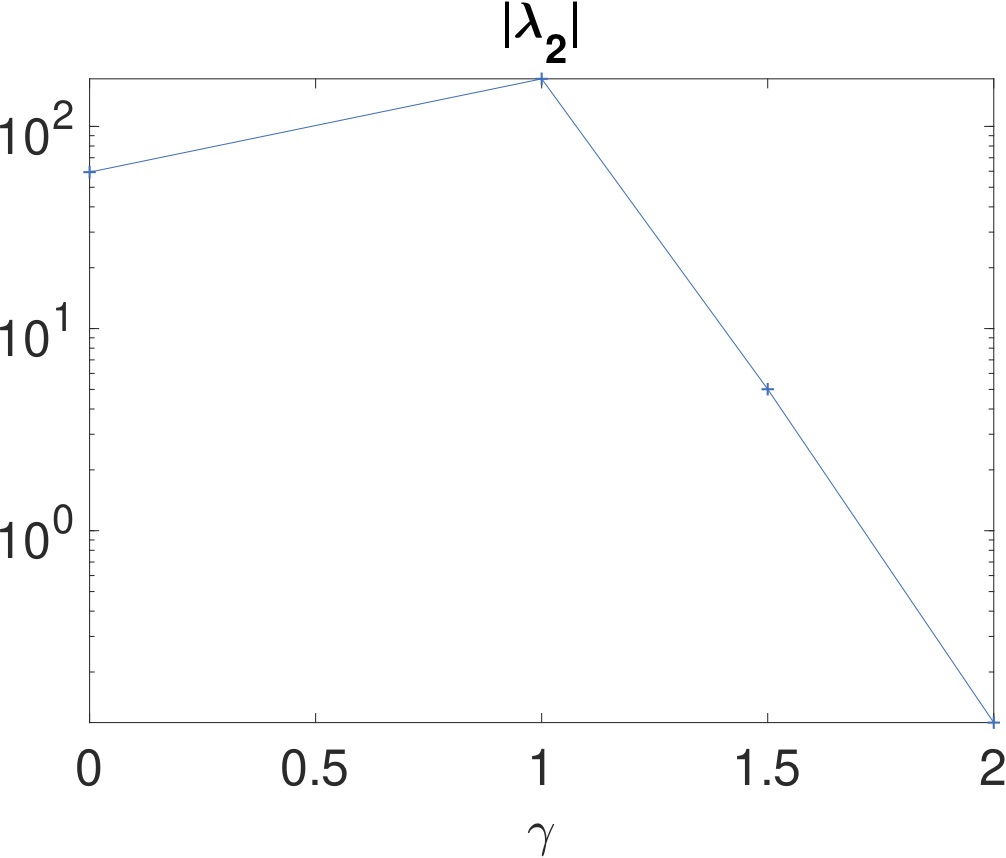}\\
    (d)&(e)&(f)
  \end{tabular}
  \caption{Unimodal case. (a) Reference measure $\mu(x)$. (b)-(e): The spectrum (part close to the
    zero) in the log-scale for the gradient flow PDEs with $\gamma=0,1,1.5,2$. (f): The smallest
    non-zero eigenvalue $\lambda_2$ for the gradient flow PDEs with different choices of $\gamma$.}
  \label{fig:nr_unimodal}
\end{figure}
We consider two simple examples in this setting. In the first example, the reference measure
$\mu(x)$ is a unimodal distribution (shown in Figure \ref{fig:nr_unimodal}(a)). Figure
\ref{fig:nr_unimodal} (b)-(e) plot the bottom part of the spectrum of the gradient flow PDEs for
$\gamma=0,1,1.5,2$. These $\gamma$ values correspond to the Pearson, KL, Hellinger, and reverse KL
divergence. We also summarize the magnitude of the smallest non-zero eigenvalue $\lambda_2$ for
these choices of $\gamma$ in Figure \ref{fig:nr_unimodal}(f). For these linear gradient flow PDEs,
$|\lambda_2|$ controls the convergence rate to the reference measure $\mu(x)$ for a generic initial
condition $\rho(t=0)$. The plot suggests that among various choices of $\gamma$, the standard
Fokker-Planck equation ($\gamma=1$) has the largest $|\lambda_2|$ and hence the fastest convergence
rate.

\begin{figure}[h!]
  \centering
  \begin{tabular}{ccc}
    \includegraphics[scale=0.1]{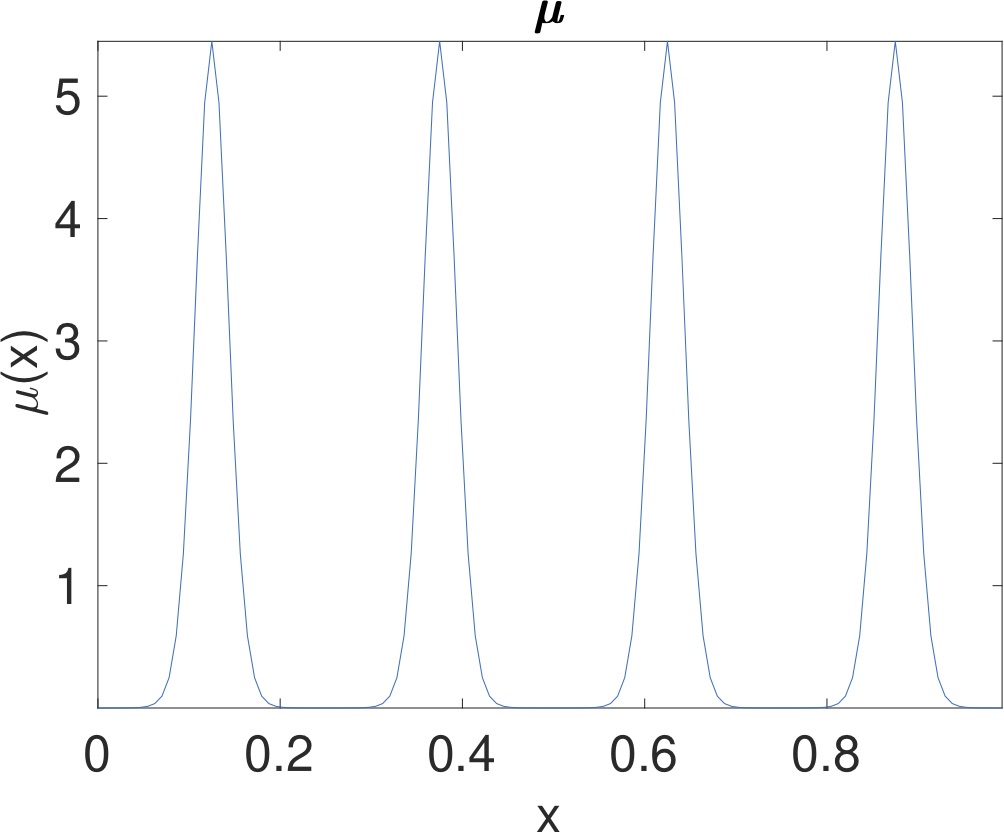}&
    \includegraphics[scale=0.1]{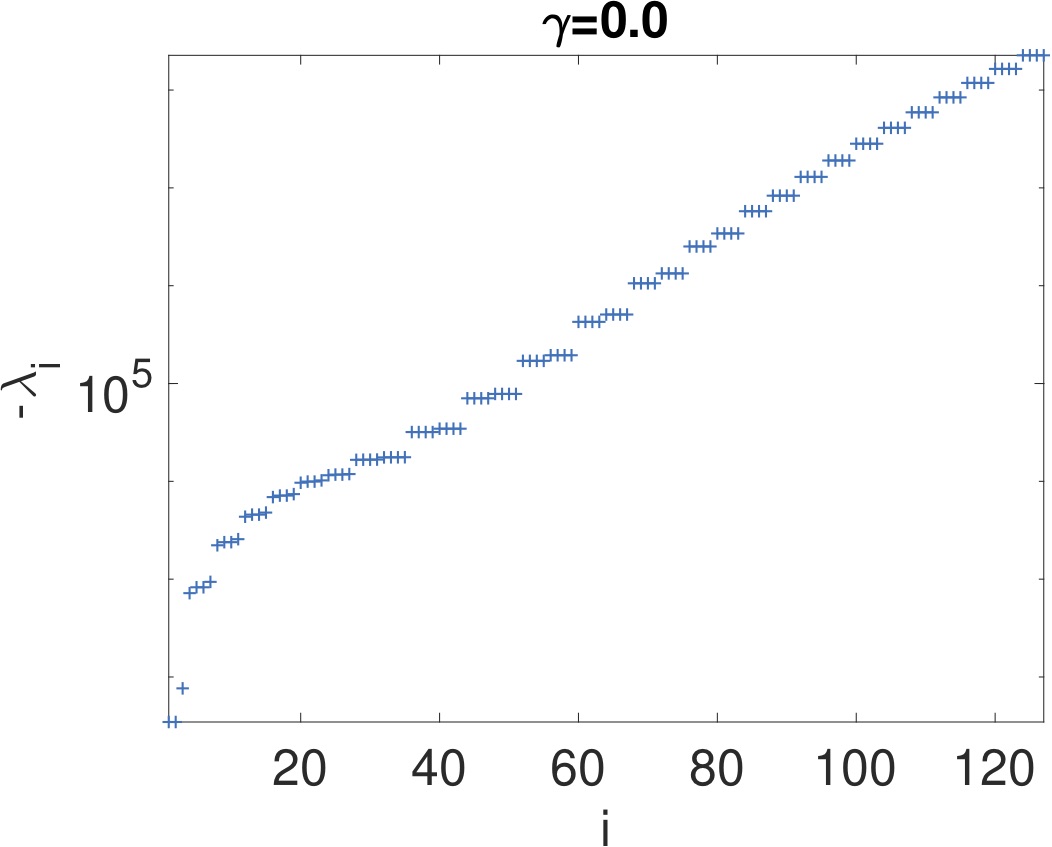}&
    \includegraphics[scale=0.1]{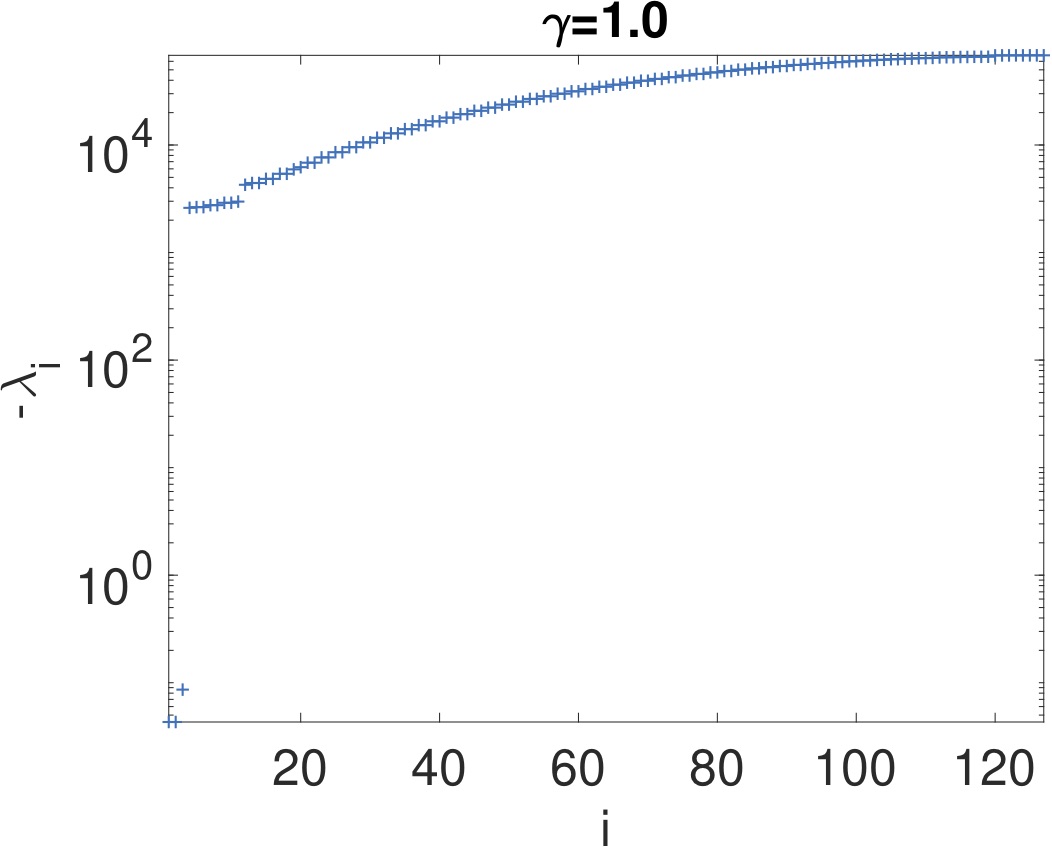}\\
    (a)&(b)&(c)\\
    \includegraphics[scale=0.1]{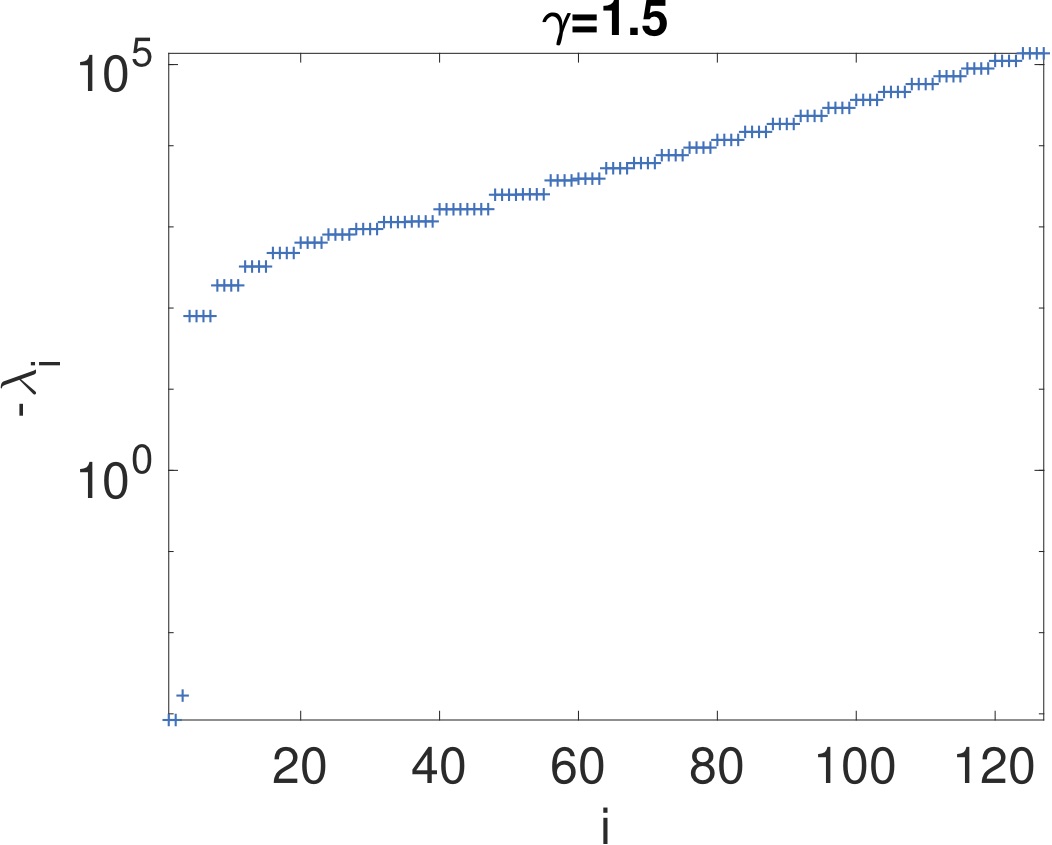}&
    \includegraphics[scale=0.1]{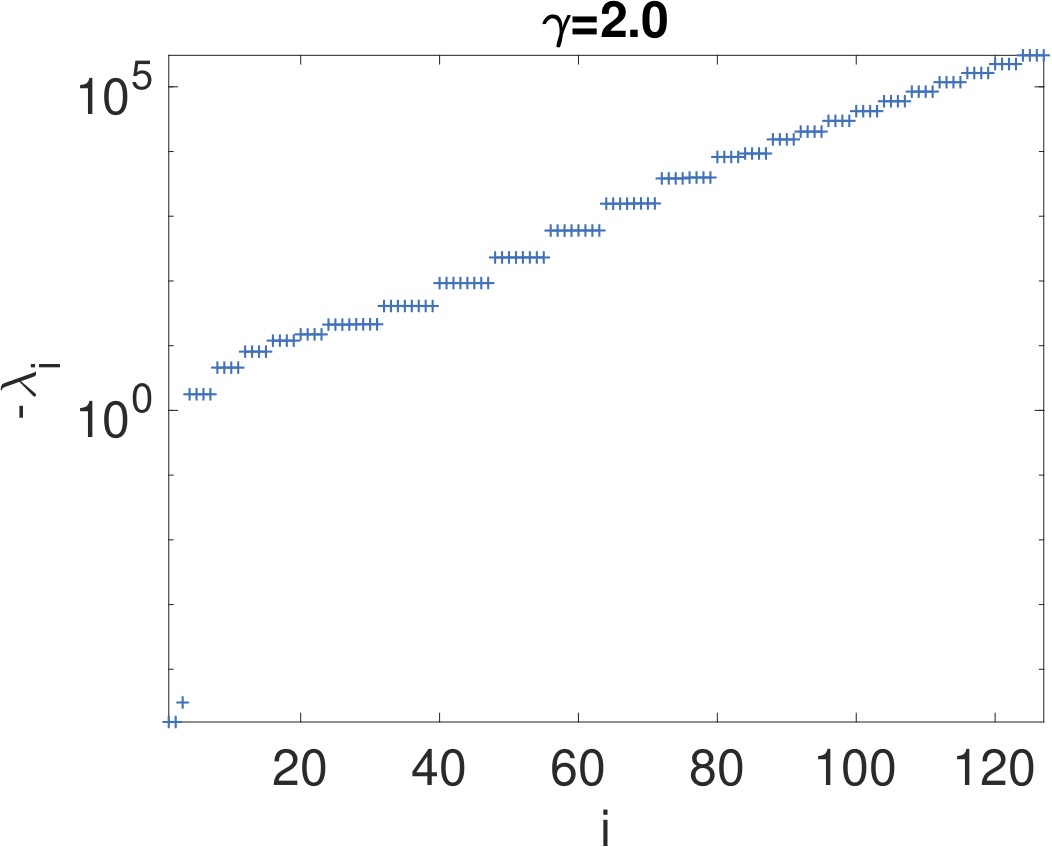}&
    \includegraphics[scale=0.1]{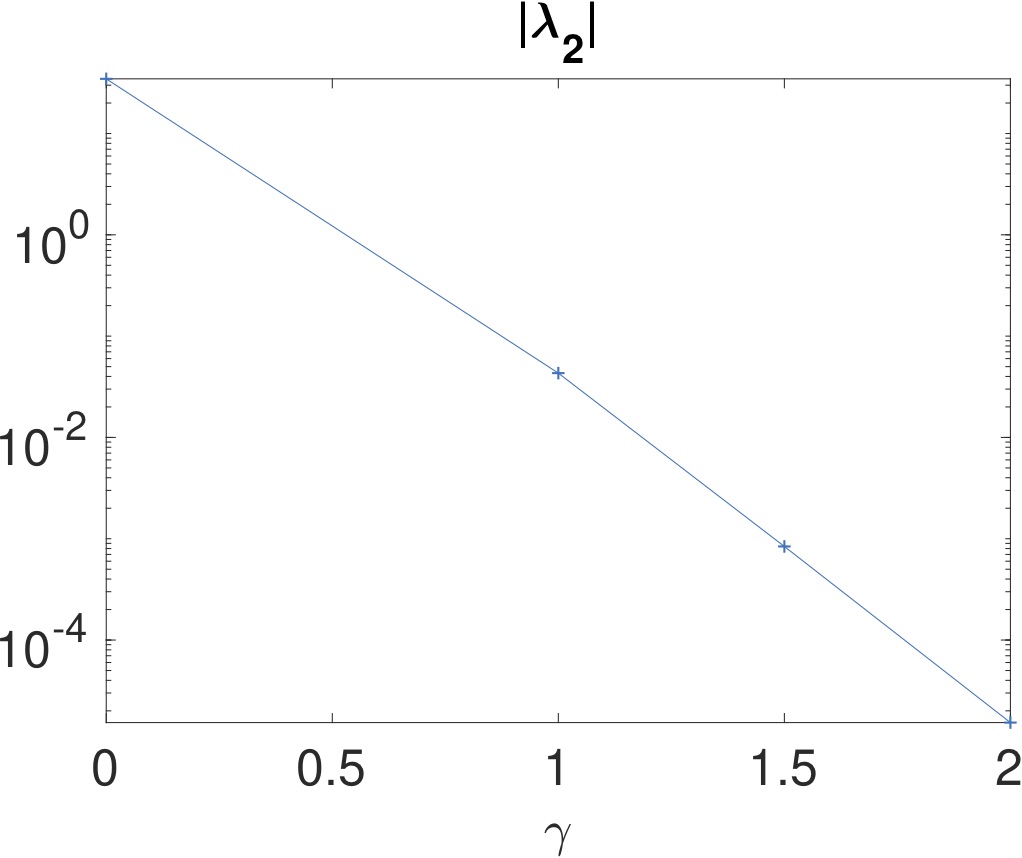}\\
    (d)&(e)&(f)
  \end{tabular}
  \caption{Multimodal case. (a) Reference measure $\mu(x)$. (b)-(e): The spectrum (part close to the
    zero) in the log-scale for the gradient flow PDEs with $\gamma=0,1,1.5,2$. (f): The smallest
    non-zero eigenvalue $\lambda_2$ for the gradient flow PDEs with different choices of $\gamma$.}
  \label{fig:nr_multimodal}
\end{figure}

In the second example, the reference measure $\mu(x)$ is a multimodal distribution (shown in Figure
\ref{fig:nr_multimodal}(a)). Figure \ref{fig:nr_multimodal} (b)-(e) plot the bottom part of the
spectrum of the gradient flow PDEs for $\gamma=0,1,1.5,2$. We again summarize the magnitude of the
smallest non-zero eigenvalue $\lambda_2$ for these choices of $\gamma$ in Figure
\ref{fig:nr_multimodal}(f). It is a well-known fact that, for the multimodal distribution, there
exists a gap between the first few lowest eigenvalues (the number of which is equal to the number of
modes) and the rest of the spectrum, due to the metastable states. For the standard Fokker-Planck
equation ($\gamma=1$), this gap is shown clearly in (Figure \ref{fig:nr_multimodal}(d)). From the
plots in Figure \ref{fig:nr_multimodal}, one can make two observations concerning the gradient flow
PDEs introduce in Section \ref{section2}. The first is that, although the gap seems to persist for
$\gamma$ greater than $1$, it decreases when $\gamma$ increases from $1$. For example in Figure
\ref{fig:nr_multimodal} the gap is significantly smaller at $\gamma=0$. The second observation is
that, in contrast to the unimodal case, $|\lambda_2|$ for the multimodal case is no longer obtained
at $\gamma=1$. In fact $|\lambda_2|$ increases quite rapidly as $\gamma$ decreases from $1$, thus
implying that the gradient flow PDE of the Pearson ($\gamma=0$) divergence converges at a faster
rate compared to the one of the standard Fokker-Planck equation ($\gamma=1$).

\section{Discussions}
In this paper, we propose a family of Riemannian metrics in the probability space, named Hessian transport metric. We demonstrate that the heat flow is the gradient flow of several energy functions under the HT-metrics. Following this, we further introduce the gradient flows of divergence functions in the HT-metrics, which can be interpreted as Kolmogorov forward equations of the associated HT-SDEs.

Our study is the first step to bridge Hessian geometry, Wasserstein geometry, and divergence functions. Several fundamental questions arise. Firstly, there are many entropies and divergence functions in information theory \cite{IG}.  Besides the $\alpha$ divergences and $\alpha$ entropy, which type of entropy's HT gradient flows of divergence functions are probability transition equations of HT-SDEs? Secondly, in machine learning applications, especially the parametric statistics, our new geometry structure leads to a new class of metrics in parameter spaces/statistical manifold. We expect some of these metrics will help the training process \cite{NP,lin2019wasserstein}.

\end{document}